%% file: trace_tree_magmas.tex
\documentclass[runningheads]{llncs}

\usepackage[T1]{fontenc}
\usepackage[utf8]{inputenc}
\usepackage{amsmath,amssymb}
\usepackage{booktabs,array}
\usepackage[hidelinks]{hyperref}

\newcommand{\Law}[1]{\ensuremath{\mathrm{E}#1}}
\newcommand{\op}{\mathbin{\diamond}}
\newcommand{\Tree}{\mathcal{T}}
\newcommand{\Code}{\mathsf{C}}
\newcommand{\Step}{\mathsf{S}}
\newcommand{\NoExc}{\mathsf{NoExc}}
\newcommand{\evalop}{\mathsf{eval}}
\newcommand{\pair}[2]{\langle #1,#2\rangle}

\newcommand{\Lean}{\textsc{Lean}}

\title{Trace-Tree Magmas: Proof-Producing Infinite Countermodels and\texorpdfstring{\\}{: }
28 New Order-Five Austin Classifications}
\titlerunning{Trace-Tree Magmas}

\author{Jiaming Zhao\inst{1,2} \and
  Bing Wu\inst{1} \and
  Tong Yang\inst{2} \and
  Xu Miao\inst{1}\thanks{Corresponding author: \email{miaoxu@zetyun.com}.}}
\authorrunning{J. Zhao et al.}
\institute{Yanbiao Lab, DataCanvas Co., Ltd., Beijing, China
  \and Peking University, Beijing, China}

\begin{document}
\maketitle

\begin{abstract}
Finite model finders cannot witness an Austin law: an identity whose finite
models are all trivial but which has a nontrivial infinite model.  We introduce
rank-decreasing sparse trace-tree magmas, a finitely presented class of total
operations on a countably infinite constructor-tree carrier.  The default
product pairs its arguments, while finitely many positive Horn clauses define
exceptional outputs.  Our main procedure derives candidate clauses from
symbolic evaluation traces.  For every reported model, it proves functionality
of the exceptional relation by descent on constructor size, proves the source
identity by exhaustive symbolic case analysis, and compiles the result into a
self-contained Lean~4 certificate.  A least simultaneous fixed point gives the
model language an implementation-independent semantics, so bounded search may
miss models but cannot invalidate a certified result.

On ETP's 96 order-five Austin candidates, we discover and Lean-verify infinite
countermodels for 28 identities that had no prior public classification in our
documented literature-and-artifact audit.  They form fourteen classes under
opposite-magma duality and establish 28 new Austin classifications; each class
has a one-clause trace-tree representative.  Supplemental
searches independently reconstruct four further identities in two duality
classes previously reported by ALPS, giving 32 certified candidates in total.
We evaluate generalization separately on Canonical-4187, the deduplicated union
of Order5-130---ten known Austin laws, the 96 candidates, and 24 identities open
on both the finite and infinite sides---and the public 4,141-row ALPS pool.  A
fresh trace run generates 636 certificates, all accepted by Judge v3: 36 from
Order5-130 and 600 from the ALPS-only canonical remainder.  In a separate
same-resource comparison at two vCPUs, 2,048~MiB, and 120 seconds per identity,
Vampire~5.0.1, E~3.5.1, and a parameter-free complete Twee~2.6.1 run have a
94-class union on the implication-proof side.  Only Twee produces trusted
counter-satisfiable outcomes: 18 canonical classes, of which independent
finite-side certificates force 16 to be infinite.  None of these ATP runs
emits an explicit model or a Lean certificate, and none is decisive on the 28
newly classified identities.  This comparison documents a complementary
capability gap.  To the best of our
audit, this is the first automated system to synthesize this trace-tree model
family, generate its well-founded inversion proofs, and emit self-contained
Lean~4 certificates.

\end{abstract}

\section{Introduction}

A magma consists of a carrier and one binary operation.  Despite this minimal
signature, its equational logic exhibits a sharp finite--infinite gap.  Austin
showed that an identity may force every finite model to be trivial while still
admitting a nontrivial infinite model~\cite{Austin1965,Austin1966}.  Such
\emph{Austin laws} are difficult for standard model search for a structural
reason: no finite multiplication table can witness the required model.

The Equational Theories Project (ETP) turns this classical phenomenon into a
large, formally specified classification problem~\cite{BolanEtAl2025}.  Its
order-five study isolated 96 identities known to have only trivial finite
models but with no known infinite model in the fixed public snapshot
~\cite{ETPOrder5Austin}.  Proving that one of these identities is an Austin law
requires more than a satisfiability answer.  One must describe an infinite
carrier and a total binary operation, prove the identity for every valuation,
and show that the carrier is nontrivial.

This paper develops automated semantic searches that produce exactly such
objects.  Their principal language consists of operations on the free
constructor algebra
\[
  \Tree ::= \mathsf e \mid \mathsf k(\Tree)
          \mid \pair{\Tree}{\Tree}.
\]
The default product is the ordered pair of its arguments.  A finite set of
universally quantified Horn rules presents exceptional products; the induced
relations are defined by a least simultaneous fixed point, without reference
to the search program.  The principal algorithm extracts rule shapes from an
evaluation trace of the source law.  Pair projections recover rule parameters,
while constructor size provides a well-founded measure for recursive
inversion.  The generated source proof then considers every relevant way in
which an earlier computation can return either an exceptional result or the
default pair.  Thus the result is not a finite approximation to a large model.
It is a finite presentation of a genuinely infinite magma.

The second defining feature is proof production.  A model is reported only
with a self-contained Lean~4 theorem that reconstructs the carrier and
operation, proves universal source validity, and refutes the target identity.
The common constructor presentation separately supplies the uniform injection
\(n\mapsto\mathsf k^n(\mathsf e)\), so every trace carrier is infinite
(Lemma~\ref{lem:carrier-infinite}).  The Python search and code generator
remain untrusted: a wrong candidate or a defective proof script is rejected
by Lean's kernel~\cite{deMouraUllrich2021,Mathlib2020}.

The trace procedure produced the first public nontrivial infinite models found
in our documented literature-and-artifact audit for 28 order-five identities,
forming 14 duality classes.  The distinction from prior ATP results is
concrete.  All 28 occur among the 96 order-five candidates that the published
Vampire campaign left unresolved; that campaign reports no constructed
infinite model for them~\cite{ETPOrder5Austin}.  In our same-resource
experiment, Vampire, E, and one parameter-free complete Twee run likewise
return no decisive status on any of the 28.  The later Vampire/E extraction of
explicit, possibly infinite models applies to the order-at-most-four ETP corpus
rather than these identities~\cite{JanotaRawsonSchulz2026}.  Our supplemental
searches also reconstruct four identities in two duality classes previously
reported by ALPS; those four are explicitly excluded from the novelty count
~\cite{ALPSData2026,XieEtAl2026ALPS}.

The published ALPS experiment is separate from our same-resource runs.  Its
eight-configuration ATP portfolio resolves 2.2\% of the 4,141-law pool at its
initial budget, and a twentyfold budget increase adds only 0.6 percentage
points.  At the largest budget it reports 110 implication proofs but only four
construction-side models, all from complete Twee; under ALPS's fixed
reasoning-model protocol, the strongest tested model solves no construction-
side instance~\cite{XieEtAl2026ALPS}.  These results expose a measured shortage
of model-construction methods rather than merely a shortage of proof-search
time.

Our novelty lies in a specific combination: a finite presentation derived
from source-law traces, a functionality proof based on rank-decreasing
inversion, and end-to-end compilation of every reported operation into a
self-contained Lean~4 theorem.  To the best of our literature and
public-artifact audit, no prior system searches this model family or generates
certificates for it.

The contributions are:
\begin{enumerate}
  \item an implementation-independent semantics for finitely presented
    trace-tree magmas, together with a trace-clause synthesis algorithm.  An
    explicit theorem chain starts from the least relational semantics, proves
    the exceptional relation functional by constructor-size induction, and
    obtains a total magma on a countably infinite carrier;
  \item a proof-producing trace search whose depth sweep, orientation schedule,
    source generalizations, and finite resource bounds are explicit, together
    with a Lean~4 compiler producing end-to-end certificates and an explicit
    trust boundary: untrusted enumeration, proof search, and text generation
    cannot create a counted result without kernel acceptance;
  \item the first public infinite models for 28 previously unclassified
    order-five identities, establishing 28 new Austin classifications in
    fourteen duality classes, together with a one-clause structural taxonomy
    and a fresh evaluation on Canonical-4187, the explicitly deduplicated union
    of Order5-130 and ALPS-4141;
  \item three supplementary searches with proof production---guarded decoding,
    strict completion, and canonical-form search---whose different model
    languages and targeted regression scopes are reported separately rather
    than merged into the comparison with Vampire, E, and Twee.
\end{enumerate}

The paper emphasizes the model language, search principles, correctness, and
mathematical structure of the synthesized models.  Runtime and memory are
reported only to establish feasibility.

Section~\ref{sec:problem} fixes the equational setting and duality notation.
Sections~\ref{sec:model-language} and~\ref{sec:search} define the model language
and the search algorithm.  Section~\ref{sec:soundness} gives the soundness
chain and trust boundary, after which Section~\ref{sec:structure} analyzes the
models found.  Section~\ref{sec:evaluation} reports the independent experiments
and the comparison with established ATPs.

\section{Problem setting}
\label{sec:problem}

\subsection{Equational nonimplication and Austin laws}

We follow the notation of the ETP report~\cite{BolanEtAl2025}.  If $X$ is an
alphabet, let $\mathcal M_X=(M_X,\diamond)$ denote the free magma generated by
$X$.  Thus a word in $M_X$ is either a letter of $X$ or a product
$u\diamond v$.  Every map $f:X\to M$ into a magma
$\mathcal M=(M,\diamond)$ extends uniquely to a homomorphism
$\varphi_f:\mathcal M_X\to\mathcal M$.  A law is a formal equality
$s\simeq t$ with $s,t\in M_X$, and
$\mathcal M\models s\simeq t$ means
$\varphi_f(s)=\varphi_f(t)$ for every $f:X\to M$.  For laws $E,F$, we
write $E\models F$ when every magma satisfying $E$ also satisfies $F$.
Following ETP, the order of a law is the total number of occurrences of
$\diamond$ on its two sides.

ETP assigns stable names \Law{n} to laws.  The law \Law{2} is
\(\mathrm{x}\simeq\mathrm{y}\) and therefore holds exactly in singleton
magmas.  A countermodel to
\[
  \Law{n}\Longrightarrow\Law{2}
\]
is consequently a nontrivial model of \Law{n}.  Combining such a model with
an independent result that every finite model of \Law{n} is trivial proves
the following classification.

\begin{definition}
A law $E$ is an \emph{Austin law} if every finite magma satisfying $E$ is
trivial and some infinite magma satisfies $E$.
\end{definition}

The certificate problem addressed here is semantic: given a source law
$E$ and target law $F$, synthesize a magma $\mathcal A=(A,\diamond)$, a proof
$\mathcal A\models E$, and a valuation refuting $F$.  For $F=\Law{2}$,
the last component is simply a pair of distinct carrier elements.

\subsection{Opposite-magma duality}

For a magma $\mathcal M=(M,\diamond)$, let
$\mathcal M^{\mathrm{op}}=(M,\diamond^{\mathrm{op}})$, where
$a\diamond^{\mathrm{op}}b=b\diamond a$.  Reversing the
two children at every internal word node maps a law $E$ to its dual $E^*$,
and
\[
 \mathcal M\models E \quad\Longleftrightarrow\quad
 \mathcal M^{\mathrm{op}}\models E^*.
\]
Finiteness, infinitude, and nontriviality are preserved.  We therefore analyze
the new results as 14 duality classes, although Lean checks a separate
certificate for every one of the 28 identities.

\section{A search language for infinite tree magmas}
\label{sec:model-language}

Three objects play different roles throughout the paper.  First, search
produces a finite rule set $R$.  Second, the rules define two relations:
$\Code_R(a,b;o)$ says that $o$ is an exceptional output allowed for the input
pair $(a,b)$, while $\Step_R(a,b;o)$ says that $o$ is an admissible intermediate
result.  An admissible result may be either the default pair or an exceptional
output.  Third, after proving that $\Code_R$ has at most one output for each
input pair, we define the total operation $\evalop_R$.  The relation $\Step_R$
is deliberately multivalued and is never used as the magma operation.  The
implementation name \texttt{Code} refers to $\Code_R$; it does not mean
generated Lean source code.

\subsection{The constructor-tree carrier}

Let $\Tree$ be the inductive set generated by a constant $\mathsf e$, a unary
constructor $\mathsf k$, and an ordered-pair constructor $\pair{-}{-}$.  Define
\[
 |\mathsf e|=0,\qquad |\mathsf k(a)|=|a|+1,\qquad
 |\pair{a}{b}|=(|a|+1)+(|b|+1).
\]
Both pair components are strictly smaller than their parent.  Constructor
disjointness and injectivity give total structural projections that reveal the
left or right component of a pair.  To make the cardinality statement
explicit, fix an injective pairing function
\(\pi:\mathbb N^2\to\mathbb N\) and define
\[
 \gamma(\mathsf e)=0,\quad
 \gamma(\mathsf k(a))=3\gamma(a)+1,\quad
 \gamma(\pair{a}{b})=3\pi(\gamma(a),\gamma(b))+2.
\]
The residue modulo three identifies the outer constructor, and injectivity of
\(\pi\) then proves by structural induction that \(\gamma\) is injective.
Conversely, $n\mapsto\mathsf k^n(\mathsf e)$ is injective by constructor
injectivity and disjointness.  Hence \(\Tree\) is countably infinite.

The carrier is chosen for three concrete reasons.  A default pair records the
two arguments of a multiplication.  Projections expose those arguments to the
generated uniqueness proof.  Constructor size turns an attempted cyclic tree
equation into an arithmetic contradiction.

\subsection{Formal clause syntax and semantics}
\label{sec:formal-language}

Let \(V=\{v_1,\ldots,v_r\}\) be clause parameters and let
\(H=\{h_1,\ldots,h_m\}\) be disjoint metavariables for intermediate results.
Write \(P(V,H)\) for the free constructor algebra generated by $V\cup H$:
\[
  p ::= v \mid h \mid \mathsf e \mid \mathsf k(p) \mid \pair{p_1}{p_2},
  \qquad v\in V,\ h\in H .
\]
A dependency-respecting trace clause is
\begin{equation}
 \forall V\,H.\quad
 \bigwedge_{i=1}^{m}\Step(p_i,q_i;h_i)
 \Longrightarrow \Code(p_L,p_R;p_O),
 \label{eq:formal-clause}
\end{equation}
The clause must satisfy three dependency conditions.  The patterns
\(p_i,q_i\) may mention only \(V\) and \(\{h_j\mid j<i\}\); the third argument
of the $i$th premise is exactly \(h_i\); and the head may mention all of
\(V\cup H\).  A ground substitution
\(\theta:V\cup H\to\Tree\) extends homomorphically to patterns.
The principal trace generator uses the distinguished source variable as
\(p_O\); the larger syntax is useful for stating the semantic class and for
the related completion procedure.

For $S\subseteq\Tree^3$, let \(\operatorname{Inst}_R(S)\) be the set of all
triples \((\theta p_L,\theta p_R,\theta p_O)\) obtained from a clause of $R$
and a ground substitution \(\theta\) whose instantiated premises all belong to
$S$.  Let
\[
 \operatorname{Raw}=\{(a,b,\pair{a}{b})\mid a,b\in\Tree\}.
\]
On the complete lattice
\(\mathcal P(\Tree^3)\times\mathcal P(\Tree^3)\), define the monotone operator
\begin{equation}
 \Phi_R(C,S)=
 \bigl(C\cup\operatorname{Inst}_R(S),
       S\cup\operatorname{Raw}\cup C\cup\operatorname{Inst}_R(S)\bigr).
 \label{eq:immediate-consequence}
\end{equation}

\begin{definition}[Relational semantics]
\label{def:relational-semantics}
For a finite dependency-respecting rule set $R$, let
\((\Code_R,\Step_R)=\operatorname{lfp}(\Phi_R)\).  Thus \(\Code_R\) is the
ternary exceptional-output relation generated by ground rule instances, and
\(\Step_R\) is the relation of allowed intermediate results generated jointly
by default pairs and exceptional outputs.
\end{definition}

\begin{proposition}[Least semantics]
\label{prop:least-semantics}
Definition~\ref{def:relational-semantics} determines a unique least pair of
relations.  It is the union of the finite iterates of \(\Phi_R\) from
\((\varnothing,\varnothing)\), and it is equivalent to the mutually inductive
default-pair, exceptional-output, and clause-instantiation rules used in the
generated Lean certificate.
\end{proposition}

\begin{proof}
The operator in Eq.~\eqref{eq:immediate-consequence} is monotone.  Its clauses
are finitary, so the union of the ascending \(\omega\)-chain from the bottom is
a fixed point.  Induction on the iteration number proves that it is contained
in every pre-fixed point, hence it is least.  The same induction translates
iterations to finite constructor derivations and conversely.
\end{proof}

When $R$ is fixed, we suppress its subscript on \(\Code_R\), \(\Step_R\),
and \(\evalop_R\).  The two constructors of the generated Lean relation are
called \texttt{raw} and \texttt{hit}; mathematically they express the following
default-pair and exceptional-output rules:
\begin{align}
  \Step(a,b;\pair{a}{b}) & &&\text{for all }a,b, \label{eq:raw-step}\\
  \Code(a,b;o) &\Longrightarrow \Step(a,b;o). \label{eq:hit-step}
\end{align}
Thus $\Step_R$ records admissible intermediate results rather than defining a
partial function.  If \(\Code_R(a,b;o)\) and
\(o\ne\pair{a}{b}\), then the same inputs have at least the two admissible
outputs $o$ and \(\pair{a}{b}\).  No theorem below assumes global
functionality of \(\Step_R\).

\begin{definition}[Well-formed trace-tree presentation]
\label{def:well-formed-presentation}
A finite dependency-respecting rule set $R$ is \emph{semantically
well formed} when its least exceptional relation is output-functional:
\[
 \forall a,b,o,o'\in\Tree.\quad
 \Code_R(a,b;o)\land\Code_R(a,b;o')\Longrightarrow o=o'.
\]
Its exceptional domain is
\(\operatorname{Exc}(R)=\{(a,b)\mid\exists o\,\Code_R(a,b;o)\}\).
The adjective \emph{finitely presented} refers only to finiteness of $R$:
\(\operatorname{Exc}(R)\) may be infinite, or even all of \(\Tree^2\).  No
set-theoretic density claim is intended.
\end{definition}

For a well-formed $R$, define
\begin{equation}
 a\mathbin{\diamond_R}b=\evalop_R(a,b)=
 \begin{cases}
   o,&\text{if }\Code_R(a,b;o)\text{ for some }o,\\
   \pair{a}{b},&\text{if no exceptional output exists.}
 \end{cases}
 \label{eq:eval}
\end{equation}
Functionality makes the first branch independent of witness choice.  The
resulting magma \(\mathcal T_R=(\Tree,\diamond_R)\) is a \emph{trace-tree
magma}.  Lean implements the first branch by classical choice and proves its
independence using the generated functionality theorem.

\begin{proposition}[Total tree operation]
\label{prop:well-defined}
If $R$ is well formed, Eq.~\eqref{eq:eval} defines a total magma on the
countably infinite carrier $\Tree$.  Moreover,
\(\Step_R(a,b;\evalop_R(a,b))\) holds for all $a,b\in\Tree$.
\end{proposition}

\begin{proof}
If no exceptional output exists, the second branch is unique and
Eq.~\eqref{eq:raw-step} supplies the step.  Otherwise all exceptional outputs agree;
the chosen output is therefore well defined and Eq.~\eqref{eq:hit-step}
supplies the step.  Countability and infinitude were proved above.
\end{proof}

The mutually inductive presentation is essential.  A rule may depend on an
earlier admissible computation without deciding whether that computation was
itself default or exceptional.  Soundness therefore cannot be inferred from
one preferred execution trace; it must cover both derivation forms of
\(\Step_R\).

\subsection{Projection--rank certificates}
\label{sec:rank-certificates}

Rank belongs to the proof discipline, not to the semantic definition above.
Intuitively, the certificate is a finite case-analysis graph: it may return to
an earlier proof state only after proving that an explicit constructor-size
measure has decreased.  The following definition makes this condition
independent of Lean syntax.  Fix a presentation $R$.
An \emph{inversion obligation} is a universally closed schematic sequent
$\Gamma\vdash G$.  Its assumptions $\Gamma$ may contain constructor
equalities, $\Code_R$- or $\Step_R$-derivations, and inequalities between
constructor sizes.  Each node of the certificate names one obligation.  An
expansion replaces that obligation by finitely many successors and records a
derivation showing that the successors imply their parent.  This derivation
uses only the first-order theory of free constructors and the usual order on
natural numbers.  Inverting a $\Code_R$ or $\Step_R$ derivation creates one
successor for each applicable constructor.  Projections or constructor clashes
may close a successor immediately.  When inversion exposes an inner
exceptional derivation, the expansion may refer back to an already named
obligation.

A \emph{projection--rank certificate} is a finite directed hypergraph of these
nodes and expansions with the following data:
\begin{enumerate}
  \item every nonrecursive expansion carries a local derivation, using only
    constructor inversion, injectivity, disjointness, projection along a fixed
    path, and valid arithmetic over the size equations;
  \item every leaf has a derivation of its obligation without successors; and
  \item for each recursive strongly connected component there is a fixed
    dimension $d$ and a rank
    \(\mu:Q_{\mathrm{ground}}\to\mathbb N^d\).  Each component of \(\mu\) is
    the size of an explicit constructor pattern in the state's ground
    arguments, and every recursive edge proves
    \(\mu(q')<_{\mathrm{lex}}\mu(q)\).
\end{enumerate}
Here $Q_{\mathrm{ground}}$ denotes the ground instances of the finitely many
schematic obligations.  A recursive edge also records the substitution that
maps the successor schema to the exposed ground subobligation.  Lexicographic
order on \(\mathbb N^d\) is well founded, so recursive back edges unfold to a
well-founded (possibly unbounded-depth) derivation rather than a finite-depth
approximation.  Certificates with no recursive edge have a vacuous rank
component.

The finite hypergraph is a metatheoretic abstraction of the generated proof,
not a separately serialized intermediate representation in the released
artifact.  The released implementation realizes only the scalar fragment
$d=1$: every mutually recursive Lean theorem uses an explicit
natural-valued constructor-size expression, such as $|a|$ or
$|\pair{o}{a}|$, in its \texttt{termination\_by} clause, and emits the
corresponding strict-decrease obligation.  It neither searches for nor emits
lexicographic tuple ranks.  The $\mathbb N^d$ formulation records a sound
extension of the proof discipline; every experimental claim in this paper and
every released certificate instantiates $d=1$.

\begin{proposition}[Certified inversion]
\label{prop:rank-certificate-sound}
Every closed projection--rank certificate proves its root sequent.  In
particular, a certificate rooted at
\[
 \Code_R(a,b;o)\land\Code_R(a,b;o')\Longrightarrow o=o'
\]
proves that $R$ is semantically well formed.
\end{proposition}

\begin{proof}
Within each recursive component, use well-founded induction on \(\mu\).  The
recorded substitution and decrease proof make every recursive successor
available at smaller rank; each nonrecursive expansion is locally sound and
each leaf is closed.  Process the acyclic component graph in reverse
topological order.  The resulting root proof is uniform over all ground
substitutions.
\end{proof}

This definition separates three notions that the implementation must not
conflate: a finite syntactic rule set, its independently defined least
semantics, and a proof that establishes the required semantic property.  In
the artifact that last object is the elaborated Lean proof term; the derivation
hypergraph above is the proof-system account used for the soundness argument.
Failure to synthesize an accepted Lean proof yields \emph{unknown}; it does not
make the rule set mathematically ill defined.

\section{Proof-producing trace search}
\label{sec:search}

The principal procedure targets a law normalized to
\(\mathrm{x}\simeq t\), where \(\mathrm{x}\) is a variable and $t$
is a non-variable term of the free magma.  This includes all 28 new order-five
classifications.  The implementation parses $t$ into an evaluation tree and
performs the following stages.  Its search state contains only the two parsed
laws, a deterministic candidate index, and resource counters; equation numbers
are absent.

\subsection{Trace-mask enumeration}

Let $N$ be the non-root internal nodes of $t$, ordered bottom-up.  A
\emph{trace mask} is a subset $M\subseteq N$.  Nodes in $M$ are called
\emph{active}; the other nodes in $N$ are \emph{inactive}.  Define a symbolic
value $r_M(n)$ recursively.  A leaf retains its variable.  An inactive node is
represented by the free pair of its two child values.  An active node $n$
receives a fresh symbol $h_n$ and contributes the premise
\[
 \Step(r_M(n_L),r_M(n_R);h_n).
\]
At the root, the two symbolic child values become the inputs of an exceptional
clause whose output is $x$.  Thus every mask determines one candidate of the form
Eq.~\eqref{eq:formal-clause}.  Masks are enumerated by increasing number of active
nodes and then deterministically by bit pattern.  This favors small semantic
descriptions, but does not enter the soundness argument.

Write $I_M=N\setminus M$ for all inactive nodes.  A node in $I_M$ is
\emph{outer-visible} if it has no active ancestor; let $V_M\subseteq I_M$ be
this subset.  The distinction matters in the universal proof.  Nodes in $V_M$
are rewritten explicitly while evaluating the root expression.  Inactive
descendants of an active node are not outer-visible, but their values may still
be needed to identify the symbolic inputs of that active node's $\Step$
premise.  The proof generator therefore reasons about every node in $I_M$, not
only $V_M$.

The released implementation exposes every finite search bound.  It accepts
source terms with two through eight internal product nodes and runs the core on
both the given and opposite orientations.  Within one core invocation, the
candidate index $i$ counts masks that survive clause construction, all three
proof builders (for functionality, universal source validity, and a target
witness), and certificate compilation.  The tuned schedule fixes $i=0$ and
instead enlarges the symbolic tableau systematically.  For each depth
\[
 d=0,1,2,3,4,5,
\]
it tries the given orientation and then the opposite orientation, allowing
9.5 seconds for each call.  Thus each call asks for the first fully supported
mask at that tableau depth; increasing $d$ permits longer chains of nested
exceptional-step inversions in absence-of-exception proofs.  This change
alters the finite enumeration schedule, not the clause semantics or acceptance
obligations.

After the depth sweep, a second finite schedule tries at most the first four
deterministic source generalizations in both orientations.  Generalization
replaces either a proper nonempty subset of repeated occurrences or one proper
compound subterm by a fresh variable.  The trace core runs at depth one and the
generated source proof explicitly specializes that variable back to the
original term.  Each generalized call has a three-second local deadline.  A
single 120-second global deadline truncates every local allowance, so late
jobs may receive less time or may not start; any call that finishes early
proceeds immediately.  These bounds restrict coverage only: every returned
candidate must satisfy the same functionality, universal-source, and
target-witness proof requirements.

For example, if the active subcomputation is
$(z\op(x\op x))\op y$, the mask may introduce
\[
 \Step(\pair{z}{\pair{x}{x}},y;h)
\]
and leave the inactive nodes as explicit free pairs.  Search never
specializes $x,y,z$ to ground trees: they remain universally quantified rule
parameters.

\subsection{Functionality and ranking synthesis}

A candidate is discarded unless the system constructs a closed
projection--rank certificate, in the precise sense of
Sec.~\ref{sec:rank-certificates}, for output functionality.  To prove that
derivations $\Code(a,b;o)$ and $\Code(a,b;o')$ have $o=o'$, the generator
combines four kinds of facts:
\begin{enumerate}
  \item constructor projections equate parameters visible at fixed paths;
  \item constructor injectivity propagates those equalities into trace
    premises;
  \item pattern-specific inversion lemmas for admissible-step premises
    recover equal clause parameters after considering the default-pair and
    exceptional-output constructors; these local lemmas do not assert global
    functionality of $\Step$; and
  \item strict size inequalities exclude cyclic or self-embedding matches.
\end{enumerate}

The rank analysis is performed symbolically.  Pair constructors generate
linear inequalities such as $|a|<|\pair{a}{b}|$.  A rule has three distinguished
values---its two inputs and its output---and hence six ordered comparisons
between distinct values.  In the general trace case, the implementation begins
with these six comparisons.  A comparison is retained exactly when, after
splitting every premise into its default and exceptional derivations, the
current retained set proves that comparison in every branch.  Iteration to a
fixed point terminates because the candidate set is finite.  The retained
comparisons become invariants of the mutually recursive inversion states;
specialized shapes use direct projection inequalities.  Each recursive call
is annotated by an explicit natural-valued constructor-size expression, and
the compiler emits the corresponding strict-decrease obligation.  Arithmetic
closes only branches whose accumulated inequalities are inconsistent.  Finite
testing of tree assignments is never accepted as a substitute for
functionality.

\subsection{Universal source-proof synthesis}
\label{sec:tableau}

Even a functional operation may fail the identity away from the trace that
suggested it.  The second filter therefore synthesizes a universal proof of
the source identity rather than replaying one preferred execution.  It builds
a finite symbolic proof tree by repeatedly splitting the constructors of
$\Step$ and $\Code$ derivations.  This proof tree is a certificate-construction
device, not an additional semantic object.  For any pair of symbolic inputs,
define the \emph{absence-of-exception predicate}
\begin{equation}
  \NoExc_R(a,b)\quad:\!\Longleftrightarrow\quad
  \neg\exists o.\,\Code_R(a,b;o).
\end{equation}
A formula $\NoExc_R(a,b)$ is an \emph{obligation} while the generator is
trying to prove it and a \emph{theorem} only after Lean has checked the
generated proof.  Its exact operational consequence is
$\evalop_R(a,b)=\pair{a}{b}$ by Eq.~\eqref{eq:eval}: no exceptional rule
applies, so the product takes the default-pair branch.

For every inactive node $n\in I_M$, with symbolic children $a_n,b_n$, the
source-proof generator constructs the obligation
\[
  \NoExc_R(a_n,b_n).
\]
It proves this formula in the context of the $\Step$ facts on which $a_n,b_n$
depend.  To do so, the generator considers both constructors of every such
fact: the default-pair
case in Eq.~\eqref{eq:raw-step} and the exceptional-output case in
Eq.~\eqref{eq:hit-step}.  In the latter case, it inverts the applicable
$\Code$ clause, uses constructor projections and unification to expose its
parameters, and may recursively analyze another exceptional-output premise.
The proof therefore does not assume that an active node took the default
branch.

A proof branch is discharged only by a contradiction that the compiler can
express and Lean can check:
\begin{enumerate}
  \item constructor projections make the required exceptional pattern
    impossible;
  \item constructor disjointness or injectivity yields a contradiction;
  \item unification would require a finite tree to contain itself; or
  \item accumulated rank inequalities form a strict cycle.
\end{enumerate}
If a branch cannot be discharged, the current mask is rejected and enumeration
continues.  For an active node, the checked predicates for its inactive
descendants identify the node's actual child values with the symbolic inputs
of its $\Step$ premise; Prop.~\ref{prop:well-defined} then supplies that
premise for the actual product.  After active subtrees have been summarized by
their fresh symbols, the predicates for the outer-visible nodes $V_M$ rewrite
the remaining products to free pairs.  The root $\Code$ clause finally rewrites
the whole right-hand side to $x$.  The generated Lean proof consists of the
corresponding case splits, equalities, projection lemmas, and arithmetic
arguments.  It covers every admissible-step case required by the trace; it is
not random testing or bounded enumeration of carrier elements.

\subsection{Witness synthesis and certificate compilation}

For target \Law{2}, the fixed witness is
$\mathsf e\ne\mathsf k(\mathsf e)$.  General targets can instead be evaluated
symbolically under a synthesized valuation.  The compiler emits the tree
datatype, $\Code$ and $\Step$, the functionality theorem, the total operation,
the universal source proof, and the target refutation required by the official
\texttt{Goal} interface.  That interface has no separate infinitude field:
infinitude of the emitted trace carrier is supplied uniformly by
Lemma~\ref{lem:carrier-infinite}.  If a proof pattern is unsupported, the
compiler returns failure rather than emitting a candidate.  The protocol
wrapper likewise rejects empty output and output beyond the submission-size
limit.  The official Judge independently enforces its proof policy and
determines semantic acceptance in Lean.

Internally, only equality chains use a typed intermediate representation.  The
released class \texttt{\_v13\_EqProof} stores the symbolic endpoints of each
chain together with its Lean text.  Before composing that text, the generator
checks the endpoints of every use of reflexivity, symmetry, transitivity,
congruence, or constructor projection.  An inconsistent chain is therefore
rejected in Python.  Functionality inversions, symbolic source-proof branches,
and target normalization are emitted as structured Lean fragments, with
separate Python checks for unsupported constructions.  The implementation does
not serialize one common proof graph for all of these tasks.  Constructor case
splits bind the actual fields of the selected $\Code$ or $\Step$ constructor
instead of asking Lean's elaborator to infer unnamed witnesses.  Both internal
representations remain untrusted: only elaboration and kernel checking of the
complete Lean theorem turn a candidate into evidence.

These failure paths affect only search completeness.  If a source-proof
branch remains unresolved, the current mask is rejected and the mask loop
continues.  If target-witness synthesis finds no separating valuation, the
current mask is likewise rejected.  After a ground target witness has been
found, failure to generate a $\NoExc_R$ proof for one of its default products
also rejects only the current mask, and enumeration proceeds to later masks.
No such failure can accept an unsupported model.

Figure~\ref{fig:trace-algorithm} gives the core search independently of
benchmark identifiers.  The implementation receives only parsed source and
target terms; equation numbers are not part of its search state or candidate
grammar.

\begin{figure}[!t]
\centering
\begin{minipage}{0.93\linewidth}
\small
\textbf{procedure} $\mathsf{TraceCore}(\mathrm{x}\simeq t,F,i,d,D)$\\
\quad $j\leftarrow0$\\
\quad \textbf{for} $M\subseteq N(t)\setminus\{\mathrm{root}\}$ in
      increasing $(|M|,\mathrm{bits}(M))$, while before $D$, \textbf{do}\\
\qquad $(R_M,I_M,V_M)\leftarrow\mathsf{BuildTraceClause}(t,M)$;
      \textbf{handle reject/abort as below}\\
\qquad $P_{\mathrm{fun}}\leftarrow
      \mathsf{EmitFunctionalityProof}(R_M)$;
      \textbf{handle reject/abort as below}\\
\qquad $P_{\mathrm{src}}\leftarrow
      \mathsf{EmitSourceProof}(\mathrm{x}\simeq t,R_M,I_M,V_M,d)$;
      \textbf{if an unresolved branch remains then reject-mask}\\
\qquad $(\nu,W_{\mathrm{tgt}})\leftarrow
      \mathsf{SynthesizeTargetWitness}(F,R_M)$;
      \textbf{if none then reject-mask}\\
\qquad $P_{\mathrm{tgt}}\leftarrow
      \mathsf{EmitTargetNormalization}(F,R_M,\nu,W_{\mathrm{tgt}})$;
      \textbf{if unsupported then reject-mask}\\
\qquad $L\leftarrow\mathsf{CompileLean}
      (R_M,P_{\mathrm{fun}},P_{\mathrm{src}},
       \nu,P_{\mathrm{tgt}})$;
      \textbf{handle reject/abort as below}\\
\qquad \textbf{if} $j<i$ \textbf{then} $j\leftarrow j+1$;
      \textbf{continue}\\
\qquad \textbf{return} $L$\\
\quad \textbf{return} \emph{failure}\\[2pt]
\textbf{where} \emph{reject-mask} continues the mask loop, whereas
      \emph{abort-job} returns \emph{failure} from $\mathsf{TraceCore}$\\[2pt]
\textbf{procedure} $\mathsf{TraceSynth}(E,F,B)$\\
\quad $D_B\leftarrow\mathsf{now}+B$\\
\quad \textbf{for each} job $J=(E_J,F_J,i,d,\delta_J,\tau_J)$ in the fixed
      schedule while $\mathsf{now}<D_B$ \textbf{do}\\
\qquad $D\leftarrow\min(D_B,\mathsf{now}+\delta_J)$\\
\qquad $L\leftarrow\mathsf{TraceCore}(E_J,F_J,i,d,D)$;
      \textbf{if} failure \textbf{then continue}\\
\qquad $L\leftarrow\tau_J(L)$\\
\qquad \textbf{if} $L$ passes the postprocessing syntax and byte-size checks
      \textbf{then return} $L$\\
\quad \textbf{return} \emph{unknown}
\end{minipage}
\caption{Core trace-model synthesis.  Here $I_M$ is the set of inactive nodes
and $V_M\subseteq I_M$ is the subset with no active ancestor.  A scheduled job
fixes an orientation or source generalization, candidate rank, proof depth,
local allowance $\delta_J$, and proof-preserving conversion $\tau_J$.  The
released depth sweep uses rank zero.  The two failure modes reproduce the
implementation: \emph{reject mask} continues enumeration, whereas
\emph{abort job} ends the current core invocation.  The returned Lean theorem
is the certificate; finite bounds affect completeness only.}
\label{fig:trace-algorithm}
\end{figure}

The algorithm maintains three acceptance obligations.  First, the generated
functionality theorem proves output uniqueness of the least relation
$\Code_{R_M}$ by constructor projection and scalar size descent.  Second, the
source theorem resolves every default-pair/exceptional-output case needed to interpret the law under
an arbitrary valuation.  Third, the target theorem evaluates both sides under
the same explicit valuation and proves distinct constructor normal forms.
The compiler may represent an obligation by a typed equality record or by a
direct Lean fragment, but it never infers validity from the mask or from a
finite test set.  The wrapper returns a model only after assembling all three
obligations into one checkable theorem.  Conversely, mask order, conservative
early termination of a core call, schedule truncation, symbolic-proof depth bounds,
and budget $B$ can omit
a valid presentation.  Returning \emph{unknown} therefore means only that no
accepted certificate was found in the explored finite search.

\subsection{Supplementary proof-producing model searches}
\label{sec:related-searches}

Three additional searches reuse parts of the constructor vocabulary and
certificate infrastructure but enumerate different mathematical
presentations.  They are not trace-mask ablations.  Their role in this paper is
supplementary: the tuned trace search receives the full Canonical-4187 union, whereas
guarded decoding, completion, and CNF are rerun on fixed historical-hit sets to
check that their optimized implementations still regenerate accepted
certificates.

\emph{Guarded decoding} enumerates a fixed grammar of total decoders on
$\mathsf e$, unary constructors, and selected pair patterns.  Multiplication
either invokes the decoder under a syntactic guard or returns a free pair.
The tuned implementation observes that its Lean proof contract forces all
seven decoder choices and therefore tries that unique admissible program
directly, in both orientations, rather than enumerating programs that the same
contract must later reject.  Shallow and depth-two tree evaluations remain
search filters; the compiler proves the resulting source identity by complete
constructor cases.  Thus the optimization changes enumeration cost, not the
accepted decoder semantics.  Section~\ref{sec:guarded-model} gives the model
for the known dual pair.

\emph{Strict completion} starts from oriented consequences of the source,
using an occurs-checked unifier and a Knuth--Bendix-style order to retain only
variable-safe, decreasing root rules.  It computes critical peaks, requires
joinability when using the confluence route, proves that every rule maps normal
inputs to a normal output, and constructs a symbolic proof of the source over
every possible choice between a matching rule and the default pair.  Selected
variants use that exhaustive proof directly rather than relying on confluence.
Like trace search, the resulting operation changes the default pairing only
at roots matched by finitely many rules, but completion derives its candidates
from critical-pair reasoning rather than trace masks.  Direct and opposite
orientations are searched explicitly.  A bounded algebraic prefix check can
discard an unpromising seed, but never serves as acceptance evidence: every
emitted candidate still requires a complete symbolic source proof and a Lean
certificate.  The main tuned implementation uses a hybrid compiler for the
same completed-rule semantics: it selects the regular proof representation for
small rule systems whose indexed proof graph contains an occurs-cycle leaf,
and the indexed representation otherwise.  The two supplemental completion
profiles used in Section~\ref{sec:evaluation} change search scheduling or proof
representation only; they do not define new model semantics.

\subsection{Canonical-form search}
\label{sec:cnf-search}

The fourth standalone program uses a narrower grammar of canonical forms.  We
retain its implementation name, \emph{CNF search}, where CNF abbreviates
\emph{canonical normal form} and is unrelated to conjunctive normal form.
Candidate seeds come from diagonal simplification, abstraction of a compound
shared by the two source sides, or direct orientation of the source.
A bounded critical-pair completion, capped at 16 rules and 16 rounds, may
extend a seed; completed systems are tried before uncompleted seeds.  A rule is
admissible only when its left-hand side has a pair at the root, its variables
are closed, and it decreases the selected order.  This pair-root condition is
enforced at orientation, seed construction, rule admission, tableau,
witness, and compilation boundaries.  It is a soundness restriction: without
it, the Python evaluator and Lean normal-form carrier can disagree about where
a rule may fire.  A retained system must also send normal inputs to normal
outputs.  Its source checker shares repeated subterms in an evaluation directed
acyclic graph and considers, at every product node, each matching root rule and
the default-pair case.
Occurs-checked unification detects cyclic substitutions; negative constraints
record when a rule pattern does not match; and normal-form closure is used to
resolve every proof branch.  Any unresolved branch rejects the system.
Target-witness search and a dedicated Lean compiler then produce the candidate.
For a retained rule system $R$, write
$\Tree_{\mathsf{NF}}(R)$ for the constructor trees containing no $R$-redex at
any subtree; output safety makes this carrier closed under the synthesized
operation.

Semantically, CNF search also defines a free-pair operation modified at roots
matched by finitely many rules.  Algorithmically it has its own seed grammar,
completion loop, symbolic source proof, compiler, and 120-second experimental
budget.  We therefore report it separately rather than folding its targeted
hits into the primary trace comparison.

\section{Correctness and the Lean trust boundary}
\label{sec:soundness}

The mathematical checker embedded in the search and the final Lean checker
serve different roles.  The former efficiently rejects bad candidates and
constructs proof terms; the latter is the authority for every reported model.

\subsection{Soundness as a theorem chain}

Fix a dependency-respecting presentation $R$, and interpret $\Code$ and
$\Step$ by the least simultaneous semantics of
Definition~\ref{def:relational-semantics}.  Write
\[
 \mathsf{Fun}(\Code)\;:\!\!\Longleftrightarrow\;
 \forall a,b,o,o'\in\Tree.\quad
 \Code(a,b;o)\wedge\Code(a,b;o')\Longrightarrow o=o'.
\]
The following statements separate semantic assumptions from the algorithm
that constructs their certificates.

\begin{lemma}[Functionality from certified inversion]
\label{lem:code-functional}
If a closed projection--rank certificate has root sequent
$\Code(a,b;o)\wedge\Code(a,b;o')\Rightarrow o=o'$, then
$\mathsf{Fun}(\Code)$.
\end{lemma}

\begin{proof}
Apply Proposition~\ref{prop:rank-certificate-sound} uniformly to every ground
quadruple $a,b,o,o'$.  Notice that the argument proves functionality of the
least exceptional relation; it does not assume, and generally cannot prove,
functionality of the multivalued relation $\Step$.
\end{proof}

\begin{lemma}[Evaluation equations and step adequacy]
\label{lem:eval-laws}
Assume $\mathsf{Fun}(\Code)$.  For all $a,b,o\in\Tree$:
\begin{enumerate}
  \item $\Code(a,b;o)$ implies $\evalop(a,b)=o$;
  \item $\neg\exists q.\,\Code(a,b;q)$ implies
        $\evalop(a,b)=\pair{a}{b}$; and
  \item $\Step(a,b;\evalop(a,b))$.
\end{enumerate}
\end{lemma}

\begin{proof}
In the exceptional case, functionality equates the witness selected in
Eq.~\eqref{eq:eval} with $o$; in the nonexceptional case the second branch of
that definition applies.  The exceptional-output and default-pair rules,
respectively, establish the third item.
\end{proof}

\begin{lemma}[Infinite carrier]
\label{lem:carrier-infinite}
The map $n\mapsto\mathsf k^n(\mathsf e)$ embeds $\mathbb N$ into $\Tree$.
Consequently every full-tree trace or guarded presentation has a countably
infinite, hence nontrivial, carrier.
\end{lemma}

\begin{proof}
Constructor injectivity and disjointness make the unary tower injective; the
encoding $\gamma$ from Section~\ref{sec:model-language} gives the converse
countability bound.
\end{proof}

\begin{lemma}[Universal trace-proof soundness]
\label{lem:tableau}
Let a mask for the law $\mathrm{x}\simeq t$ determine a well-formed rule set
$R$, the inactive-node set $I_M$, and the outer-visible subset $V_M$.  Suppose
the symbolic proof establishes $\NoExc_R(a_n,b_n)$ for every $n\in I_M$ under
both the default-pair and exceptional-output cases of every $\Step$ premise on
which $a_n,b_n$ depend.  Then
$\mathcal T_R\models\mathrm{x}\simeq t$.
\end{lemma}

\begin{proof}
Fix an arbitrary valuation $f:X\to\Tree$ and evaluate $t$ from the leaves
upward.  At every inactive node, its $\NoExc_R$ theorem and
Eq.~\eqref{eq:eval} identify the actual product with the free pair used in the
symbolic trace.  This has two distinct uses.  First, inactive descendants of
an active node identify its actual child values with the symbolic inputs in
that node's $\Step$ premise.  Lemma~\ref{lem:eval-laws} then supplies the
premise for the active node's actual output, whether that output is the default
pair or an exceptional value.  Continuing upward establishes every active-node
premise and instantiates its fresh result symbol.  Second, once active subtrees
have been summarized in this way, the $\NoExc_R$ theorems for the
outer-visible nodes $V_M$ rewrite the remaining products above them to their
symbolic free pairs.  The two root inputs now match the instantiated head of
the generated $\Code$ clause, whose output is $f(\mathrm{x})$.
Lemma~\ref{lem:eval-laws} therefore identifies the root product with
$f(\mathrm{x})$.  Since $f$ was arbitrary, the law holds.
\end{proof}

\begin{lemma}[Countervaluation soundness]
\label{lem:countervaluation}
Let $F$ be $u\simeq v$.  If, under one valuation $\nu$, symbolic evaluation
derives $\varphi_\nu(u)=a$, $\varphi_\nu(v)=b$, and $a\ne b$, then
$\mathcal T_R\not\models F$.
\end{lemma}

\begin{proof}
The two evaluation derivations interpret the same operation and valuation.
Substituting them into a hypothetical equality for $F$ gives $a=b$, contrary
to the constructor inequality.
\end{proof}

\begin{lemma}[Soundness of scheduled transports]
\label{lem:transport}
Every transport $\tau_J$ used by the fixed schedule preserves a certified
nonimplication and the cardinality of its carrier.
\end{lemma}

\begin{proof}
The schedule composes only four operations.  Swapping the two sides of a law
uses symmetry of equality.  A generalized source $E_J$ is accompanied by an
explicit term substitution $\sigma$ with $\sigma(E_J)=E$; satisfaction of a
universal identity is closed under term substitution, so a model of $E_J$ is
a model of $E$.  Identity jobs require no transport.  Finally, an opposite
job applies the equivalence of Section~\ref{sec:problem} simultaneously to
the source and target: a model of $E^*$ refuting $F^*$ yields its opposite
magma, which models $E$ and refutes $F$.  None of these operations changes
the underlying carrier.  Closure under composition proves the claim.
\end{proof}

\begin{theorem}[Soundness of a kernel-accepted trace-search output]
\label{thm:candidate}
Suppose the procedure in Fig.~\ref{fig:trace-algorithm} returns a Lean
declaration for source $\mathrm{x}\simeq t$ and target $F$, and the fixed
official Lean environment accepts that declaration under the stated proof
policy.  Then its
presentation defines a countably infinite magma satisfying
$\mathrm{x}\simeq t$ and refuting $F$.
\end{theorem}

\begin{proof}
Kernel acceptance checks the generated functionality, source, and target
proofs rather than trusting the Python emitter.  The checked functionality
proof and Lemma~\ref{lem:code-functional} make $R$ well formed;
Proposition~\ref{prop:well-defined} supplies the total operation.
Lemmas~\ref{lem:carrier-infinite}, \ref{lem:tableau}, and
\ref{lem:countervaluation} establish infinitude, source validity, and target
failure for the core search instance, respectively.  Lemma~\ref{lem:transport}
transfers these conclusions to the source and target supplied to the outer
procedure.
\end{proof}

\begin{corollary}[Kernel-accepted certificate]
\label{cor:kernel-result}
If the fixed official Lean environment accepts a generated declaration of
type \texttt{Goal}, then the submitted carrier and operation satisfy the
encoded source theorem and refute the encoded target.  Combining that checked
declaration with Lemma~\ref{lem:carrier-infinite} for full-tree carriers, or
Proposition~\ref{prop:nf-soundness} for normal-form carriers, yields the
claimed infinite countermodel.
\end{corollary}

\begin{proof}
The declaration packages the selected carrier and operation, the universal
source theorem, and a target countervaluation against the official problem
interface.  Lean elaboration and kernel checking reconstruct those claims.
The separate theorem about the chosen carrier supplies infinitude, which is not a
separate field of the official \texttt{Goal} interface.
\end{proof}

Each returned file contains a complete Lean instance of the relevant argument.
For a trace model, the file proves $\Code$ functional before defining
$\evalop$ by choice.
In particular, no certificate assumes that $\Step$ is globally functional:
its local inversion lemmas are used only under the constructor patterns and
premise context for which they were generated.
For the normal-form specialization of Sec.~\ref{sec:normal-form}, both the
completion and CNF compilers prove closure of the carrier under every possible
selected rule.  Their symbolic source proofs consider every root rule and the
default-pair case at each product node.  They use occurs-checked unification,
record negative constraints when a rule does not match, and reject a candidate
if any proof branch remains unresolved.
Source theorems in both model families quantify over all carrier variables,
and the final declaration inhabits the official
\texttt{Goal} type with the carrier, magma instance, source theorem, and target
refutation.

\begin{proposition}[Normal-form soundness]
\label{prop:nf-soundness}
Let $R$ be a finite set of root rules whose outputs are normal whenever their
inputs are normal.  Suppose the symbolic proof resolves every default-pair and
rule branch of the source, and an explicit normal valuation separates the two
target terms.  Then the choice operation induced by $R$ on
$\Tree_{\mathsf{NF}}(R)$ is an infinite model of the source and a countermodel
to the target.
\end{proposition}

\begin{proof}
Closure makes the operation total on $\Tree_{\mathsf{NF}}(R)$.  The symbolic
proof covers every output that the choice operation can select at each source
node, so every evaluation satisfies the source.  The unary tower is normal and
injective, hence the carrier is infinite; the recorded valuation refutes the
target.  Both compilers check closure, the universal source theorem, and the
target refutation in Lean.  The CNF compiler also emits its tower-injectivity
lemma; for completion certificates the same carrier fact is supplied by the
uniform mathematical argument above.
\end{proof}

\subsection{Trust boundary}

\begin{figure}[!htbp]
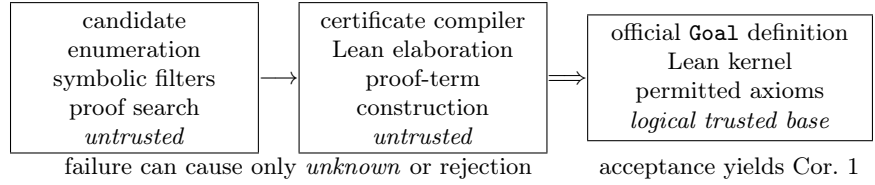

\centering
\small
\setlength{\tabcolsep}{3pt}
\begin{tabular}{@{}c@{\(\longrightarrow\)}c@{\(\Longrightarrow\)}c@{}}
\fbox{\parbox{0.25\linewidth}{\centering
 candidate enumeration\\symbolic filters\\proof search\\\emph{untrusted}}}
&
\fbox{\parbox{0.25\linewidth}{\centering
 certificate compiler\\Lean elaboration\\proof-term construction\\\emph{untrusted}}}
&
\fbox{\parbox{0.29\linewidth}{\centering
 official \texttt{Goal} definition\\Lean kernel\\permitted axioms\\
 \emph{logical trusted base}}}
\\[4pt]
\multicolumn{2}{c}{failure can cause only \emph{unknown} or rejection}
&
acceptance yields Cor.~\ref{cor:kernel-result}
\end{tabular}
\caption{Data flow and trust boundary.  Search metadata and Python proof
objects are hints only; the counted result is the theorem reconstructed from
the submitted source and accepted by the official Lean environment.}
\label{fig:trust-boundary}
\end{figure}

The logical trusted base is the Lean kernel together with the official
\texttt{Goal} definition and the axioms permitted by the Judge policy.  The
elaborator translates source into proof terms but does not enlarge that base:
the kernel checks the resulting term and all imported Mathlib declarations.
Enumeration, heuristics, Python unification, and certificate text generation
are likewise untrusted.  They may miss a model or emit a rejected file, but,
relative to the official encoding and permitted axioms, cannot make a false
theorem pass kernel checking.  Figure~\ref{fig:trust-boundary} also separates
logical validation from experimental bookkeeping: a JSON status or a
generator's own \texttt{valid} flag is never counted without the matching
accepted Judge receipt.  A receipt records the pinned checking environment;
the theorem's logical authority comes from kernel checking, not from the JSON
record or its archive hash.

\section{Mathematical structure of the synthesized models}
\label{sec:structure}

The 28 new classifications are not 28 unrelated programs.  Modulo duality,
all fourteen use the full carrier $\Tree$, the same free-pair default, the same
constructor size, and exactly one constructor for $\Code$.  For a presentation
of the form Eq.~\eqref{eq:formal-clause}, call $m$ its \emph{trace arity} and
call the ordered dependency graph
$h_j\to h_i$ whenever $h_j$ occurs in the $i$th premise its \emph{trace
dependency graph}.  Dependency-respecting syntax makes this graph acyclic.
Eleven discovered classes have trace arity one; three have trace arity two.
Table~\ref{tab:trace-schemas} gives their complete mathematical presentations.

\input{generated/trace_schemas_table.tex}

This yields a compact structural classification.

\begin{proposition}[One-clause structure]
Every new Austin model reported here is a well-formed, finitely presented
one-clause extension of the free-pair magma, and its functionality is witnessed
by a projection--rank certificate.  Up to opposite-magma duality, its complete
law-specific data consist of one Horn clause with at most two admissible-step
premises.
\end{proposition}

\begin{proof}
For each of the fourteen rows of Table~\ref{tab:trace-schemas}, the archived
certificate defines precisely the displayed clause as the sole constructor of
$\Code$, proves the corresponding \texttt{code\_unique} theorem, and proves
the source theorem by complete default-pair/exceptional-output case analysis.
Counting the displayed premises
gives eleven rows of arity one and three rows of arity two.  The certificate for
the dual law transports the operation to the opposite magma and is checked
separately.
\end{proof}

\begin{theorem}[Twenty-eight new Austin classifications]
\label{thm:austin-classifications}
Each of the first 28 laws in Table~\ref{tab:equations} is an Austin law.  They
form fourteen classes under opposite-magma duality.
\end{theorem}

\begin{proof}
The finite-model half is the published ETP classification: every finite magma
satisfying any of these 28 laws is trivial~\cite{ETPOrder5Austin}.  For one
representative of each row of Table~\ref{tab:trace-schemas}, the archived Lean
certificate checks the universal source theorem and an explicit nontrivial
valuation in the full-tree carrier.  Lemma~\ref{lem:carrier-infinite} proves
that this carrier is infinite.  The separately checked opposite construction
and the equivalence in Section~\ref{sec:problem} give the dual representative.
Thus both clauses of the Austin-law definition hold for all 28 laws.
\end{proof}

\subsection{Worked example: a one-premise trace model}

Consider \Law{7755}:
\[
 x\simeq y\op\bigl(y\op((z\op(x\op x))\op y)\bigr).
\]
Its only exceptional clause is
\begin{equation}
 \Step(\pair{z}{\pair{x}{x}},y;h)
 \Longrightarrow \Code(y,\pair{y}{h};x).
 \label{eq:7755-clause}
\end{equation}
To match Lemma~\ref{lem:tableau} explicitly, name the five product nodes from
the leaves upward:
\[
 n_0=x\op x,\quad n_1=z\op n_0,\quad n_2=n_1\op y,
 \quad n_3=y\op n_2,\quad n_4=y\op n_3.
\]
The mask makes $n_2$ active and uses the root $n_4$ as the exceptional-clause
head.  The nodes $n_0,n_1,n_3$ are inactive.  The generated proofs of
$\NoExc_R$ for $n_0$ and $n_1$ show that their actual values are
$\pair{x}{x}$ and $\pair{z}{\pair{x}{x}}$; these equalities identify the first
input of the active-node premise.  Proposition~\ref{prop:well-defined} then
supplies
\[
 \Step(\pair{z}{\pair{x}{x}},y;h)
\]
for the actual value $h$ of $n_2$.  The $\NoExc_R$ proof for $n_3$, which is
checked under both possible constructors of that $\Step$ premise, yields
$n_3=\pair{y}{h}$.  Equation~\eqref{eq:7755-clause} therefore makes the root
$n_4$ equal to $x$.  This is precisely the two-part use of inactive nodes in
Lemma~\ref{lem:tableau}: $n_0,n_1$ justify the symbolic inputs of an active
node, while the outer-visible node $n_3$ rewrites the expression above that
active node.

Functionality is equally structural.  The left projection of the exceptional
clause's second input recovers its first input $y$.  If two instances of
Eq.~\eqref{eq:7755-clause} share both inputs, constructor injectivity equates
their trace outputs.  The generated pattern-specific $\Step$-inversion lemma
then recovers the embedded $x$ values.  Any remaining cyclic case would force
a finite constructor tree to contain a strictly larger copy of itself, and the
size inequalities rule it out.
The opposite operation gives the dual model for \Law{38249}.

\subsection{Two-premise models}

Three exceptional classes remember two earlier evaluations rather than one.
For \Law{38303}, the clause has the sequential trace
\[
 \Step(x,z;h_0)\wedge
 \Step(y,\pair{h_0}{x};h_1)
 \Longrightarrow \Code(\pair{h_1}{y},y;x).
\]
The second premise consumes the first trace output, and the root pattern
recovers the second.  Projection and size reasoning then propagates uniqueness
back through both premises.  No new carrier or semantic principle is needed.

The newly found class \Law{9680}/\Law{36524} instead records two independent
steps:
\[
 \Step(z,y;h_0)\wedge \Step(x,y;h_1)
 \Longrightarrow \Code(y,\pair{h_0}{\pair{y}{h_1}};x).
\]
Thus trace arity two does not by itself imply a dependency edge: the two
premises can be evaluated independently and only meet in the exceptional root
pattern.  The third arity-two class, \Law{35036}/\Law{11081}, has the same
empty dependency graph with a different root projection pattern (Table
~\ref{tab:trace-schemas}).

\subsection{The previously classified guarded pair}
\label{sec:guarded-model}

The two guarded-decoding certificates, \Law{22619}/\Law{22634}, independently
formalize a pair already classified by ALPS.  They use the same tree carrier
but a guarded decoder $\delta$:
\[
\begin{aligned}
 \delta(\mathsf e)&=\mathsf e,&
 \delta(\mathsf k(a))&=a,\\
 \delta(\pair{a}{a})&=\mathsf k(a),&
 \delta(\pair{a}{\pair{a}{x}})&=x,
\end{aligned}
\]
with all other pair shapes decoded to $\mathsf e$.  For \Law{22619},
multiplication returns $\delta(a)$ when the right input is $\mathsf e$.  When
$b\ne\mathsf e$ and
$a=\pair{b}{b}$, it returns $\mathsf e$; in every other case it returns the
free pair $\pair{a}{b}$.  The key identity
$\delta(y\op(y\op x))=x$ supplies the required cancellation.  The model for
\Law{22634} uses the opposite operation.  These two certificates strengthen
the formal evidence but are excluded from the novelty count.

\subsection{Normal-form models from completion and CNF search}
\label{sec:normal-form}

The completion procedure and the independently extracted CNF procedure
use different candidate enumerators and Lean compilers, but their accepted
models share the following mathematical form.  They synthesize finitely many
root rules
\[
  \Code(\ell_i,r_i;o_i),\qquad i=1,\ldots,k,
\]
over constructor patterns.  Let $\mathsf{NF}_R(t)$ mean that every subtree is
normal and no pair root matches a rule input.  The carrier becomes
\[
  \Tree_{\mathsf{NF}}(R)=\{t\in\Tree\mid\mathsf{NF}_R(t)\}.
\]
On normal inputs, multiplication chooses a matching root-rule output when one
exists and otherwise forms a free pair.  Search proves that every possible
rule output is normal and symbolically checks the source under every matching
rule and default-pair case.  Overlapping rules need not be confluent: classical choice selects
one output, while the proof covers every selectable constructor.  The unary
tower remains normal and therefore proves infinitude.

For \Law{20034},
\[
 x\simeq(y\op y)\op((z\op(x\op x))\op z),
\]
one root rule suffices:
\[
 \pair{y}{y}\op\pair{\pair{z}{\pair{x}{x}}}{z}\longmapsto x.
\]
Normality forces the displayed inner products to be free pairs and the source
contracts at the root.  In the clean experiment, completion and CNF each
independently generated and certified this same one-rule presentation.  Across
the six further duality classes in Table~\ref{tab:additional-models}, accepted
CNF presentations use one to thirteen root rules; completion covers five of
the six classes.  CNF also reconstructs the previously published ALPS pair
\Law{19966}/\Law{26105}.  These results illustrate a complementary model
language, but neither they nor the guarded pair enter the new-classification
count.  Agreement on a model does not identify the two enumeration algorithms.

\section{Evaluation}
\label{sec:evaluation}

The evaluation addresses four questions.  First, does each generated object
survive independent Lean checking?  Second, what does each search contribute
under its stated evaluation scope?  Third, does trace synthesis transfer beyond
the Austin identities on which it was developed?  Fourth, what capability
remains uncovered by established equational ATPs under the same resource tier?
Trace is the only procedure evaluated on the full Canonical-4187 union and is the only
one compared with the ATPs.  Guarded decoding, completion, and CNF search are
reported as targeted regressions of complementary model languages.

\subsection{Two source groups and two deduplication levels}

We froze two evaluation groups and their SHA-256 hashes before running the
procedures; neither contents nor hashes were changed in response to outcomes.
The first group, Order5-130, is a disjoint union of three order-five
subsets: KnownAustin-10, ten previously known Austin laws;
Candidate-96, the 96 identities listed by ETP as having only trivial
finite models and, in the source snapshot, no known infinite model; and
Open-24, 24 identities for which both finite- and infinite-model statuses
were unresolved in the originating study.  The second group,
ALPS-4141, is the public 4,141-row ALPS evaluation pool, spanning orders
five through eight~\cite{ALPSData2026,XieEtAl2026ALPS}.
For backward compatibility, machine-readable provenance retains the legacy
tags \texttt{austin96} for Candidate-96 and \texttt{generality34} for
KnownAustin-10 together with Open-24.  These are source tags, not additional
evaluation groups.

Together the two groups contain 4,271 source rows: 130 in Order5-130 and 4,141
in ALPS-4141.  Of the Order5 identities, 83 also occur as ALPS rows.  Retaining one record
for each cross-group overlap gives Merged-4188: all 130 Order5 records
and 4,058 ALPS-only records.  At a second, explicitly different level, we
quotient Merged-4188 by variable renaming and exchange of the two sides of an
equality.  Exactly one additional pair, internal to ALPS-4141, collapses under
this equivalence, giving Canonical-4187: 130 Order5 classes and 4,057
ALPS-only classes.  Magma duality is \emph{not} part of this canonicalization;
we use it only to organize mathematical discoveries.

Unless a sentence explicitly says ``source row'' or ``original ALPS row,'' all
full-search and ATP counts below refer to Canonical-4187.  The internal ALPS
duplicate belongs to a class solved by trace.  Consequently, 624 accepted
canonical classes have an ALPS representative, but they project to 625
accepted rows in the original 4,141-row file.  This convention prevents counts
from the two overlapping source groups from being added as though they were
disjoint.

The tuned trace executable was run afresh on all 4,187 classes in
Canonical-4187.
The other executables were rerun only on frozen historical-hit sets: two
problems for guarded decoding, 61 for completion, and 17 for sound CNF search.
These regressions test preservation of known capability and certificate
validity; they are not estimates of coverage on Canonical-4187.  Every profile has
a 120-second internal search budget, a 125-second outer allowance, a
2,000-MiB process-tree ceiling, and a 1,024,000-byte certificate ceiling in a
two-vCPU, 2,048-MiB sandbox.  Search output is counted only when the exact Lean
source emitted for that run has an \texttt{accepted} Judge v3 receipt.

\subsection{Independent coverage and certificate validity}

\paragraph{Primary trace result.}
\input{generated/trace_full_results_table.tex}

The tuned depth sweep generates 636 certificates, and Judge v3 accepts all
636.  Among the 3,551 unsolved problems, 3,295 terminate without a candidate
and 256 reach the process-tree memory ceiling.  These are outcomes of a bounded
search, not counterevidence to the existence of a model.  The largest emitted
certificate is 521,804 bytes, below the configured ceiling.

On Candidate-96, trace generates and certifies 28 new infinite countermodels, or
fourteen classes under duality.  The other 68 problems are unresolved by this
bounded trace run: 65 terminate without a candidate and three reach the memory
ceiling.  All 28 identities occur in the published 96-candidate campaign and
were left unresolved by its Vampire procedure~\cite{ETPOrder5Austin}.  In the
released ALPS data, 24 have null baseline resolution; \Law{4957}, \Law{5012},
\Law{40917}, and \Law{41252} are absent from the screened pool.  The exact-law
literature and artifact audit found no earlier classification for any of the
28.

Across Order5-130, trace accepts 36 identities.  Its three constituent subsets
remain separate in Table~\ref{tab:trace-full-results}: Candidate-96 contributes
28; KnownAustin-10 contributes six, namely three known Austin duality classes;
and Open-24 contributes the pair \Law{17260}/\Law{28740}.  The last pair now
has certified infinite models but is not called Austin until its finite side is
settled.  On the 4,057 ALPS-only classes, trace accepts 600.  A further 24
accepted Candidate-96 identities also occur in ALPS-4141, so 624 accepted
canonical classes have an ALPS representative.  The accepted internal ALPS
duplicate makes these 624 classes correspond to 625 rows of the original pool.
This transfer result is evidence that the trace language is not an encoding of
the fourteen new classes by equation identifier.

\paragraph{Supplementary searches.}
\input{generated/isolated_results_table.tex}

Table~\ref{tab:isolated-results} deliberately reports the three other search
families on their historical-hit regressions.  Guarded decoding regenerates
and certifies \Law{22619}/\Law{22634}.  The main completion profile generates
60 of its 61 targets: Judge accepts 57, while three receive no terminal response
before the experiment's 60-minute client cutoff.  It does not generate
\Law{9001325}.  A generalized schedule accepts that last problem, and an
indexed certificate profile accepts the three client-cutoff cases.  These
alternatives cover the same completed-rule model language; they differ in
candidate schedule and Lean proof representation.  We therefore count 61
completion problems, not 65 certificates.  Sound CNF search regenerates and
certifies all 17 of its targets.

Across the four selected accepted sets, trace, guarded decoding, completion,
and CNF cover 668 canonical classes in union.  Viewed through the two source
groups, this union contains 44 Order5-130 classes and 652 classes having an
ALPS-4141 representative (653 original ALPS rows); these views overlap and
must not be added.  The individual selected sets have sizes 636, 2, 61, and 17.
Trace contributes 598 classes not covered by a selected supplementary set;
guarded contributes two, completion 21, and CNF seven.  The intersections are
38 for trace--completion, eight for trace--CNF, and ten for completion--CNF;
eight problems lie in all three of those sets, and guarded is disjoint from
them.  Because only trace was run on all of Canonical-4187, this 668-class
union is coverage accounting for the accepted artifacts, not a full-benchmark
accuracy estimate for the supplementary searches.

The order-five supplemental results clarify what the other model languages
add.  Guarded decoding reconstructs the ALPS pair
\Law{22619}/\Law{22634}; CNF reconstructs the ALPS pair
\Law{19966}/\Law{26105}.  These four identities were already classified and
are excluded from the 28-new-identity count.  Within the KnownAustin-10 and
Open-24 subsets of Order5-130, completion covers five of the six classes in
Table~\ref{tab:additional-models}, while CNF
covers all six.  Trace covers four.  The Open-24 pair
\Law{17260}/\Law{28740} remains described only as an infinite countermodel.

\input{generated/additional_models_table.tex}

Across the public result table there are 731 certificate--profile occurrences:
728 accepted and the three completion client-cutoff records above.  This total
includes alternate certificates, cross-family overlap, and twelve archived
trace certificates from an earlier preflight; it is not a solved-problem
count.  There are no incorrect, malformed, or other non-timeout Lean
rejections.  An earlier trace preflight solved 643 problems, but its later
repeat missed twelve of them; their accepted certificates are retained as a
separately labelled fluctuation archive and are excluded from the fresh count,
the primary conclusions, and Table~\ref{tab:trace-full-results}.

\subsection{Construction-aware ATP comparison}

\paragraph{Our same-resource experiment.}
We ran Vampire~5.0.1~\cite{KovacsVoronkov2013}, E~3.5.1
~\cite{Schulz2002}, and Twee~2.6.1~\cite{Smallbone2021} independently on both
frozen source groups.  Every tool received one 120-second allowance per source
row in a two-vCPU, 2,048-MiB sandbox, with a sampled 2,000-MiB aggregate
process-group RSS limit.  Vampire used its CASC portfolio and E used
\texttt{--auto}, both in the conjecture-proving direction.  Twee used one
parameter-free complete run, \texttt{twee --tstp --formal-proof problem.p}; the
last option retains proof output but does not change the search.  This is one
ALPS-aligned completion configuration, not the published eight-configuration
ALPS sweep.  Table~\ref{tab:same-resource-atp} keeps Order5-130 and ALPS-4141
separate because 83 Order5 rows recur in the ALPS pool.  We compare the ATPs
with trace alone, the only one of our searches run on the full canonical union.

\input{generated/same_resource_atp_table.tex}

The result types answer different questions.  An implication proof establishes
that the source law forces \Law{2}; it is a proof-side result, not a model.
Raw \texttt{CounterSatisfiable} messages from the incomplete Vampire and E
proof-mode runs are retained in the artifact as diagnostics but are not treated
as scientific verdicts.  Only a clean terminating saturation from the
parameter-free complete Twee run is counted as trusted counter-satisfiable.
Such a result establishes a nontrivial model of the source law, but does not by
itself identify its carrier, prove it infinite, or provide a Lean certificate.

After cross-source deduplication and quotienting by variable renaming and
equality symmetry, the three-tool union on Canonical-4187 contains 94
implication proofs and 18 trusted counter-satisfiable classes, with no conflict.
Independent finite-side certificates exclude nontrivial finite models for 16
of those 18 classes, so infinitude follows for them; the remaining pair,
\Law{17260}/\Law{28740}, belongs to Open-24 and yields only nontrivial-model
existence.  No ATP run emits an explicit carrier operation or a Lean-checkable
model certificate.  By contrast, trace constructs and Lean-certifies explicit
infinite magmas for 636 canonical classes.  Its accepted set is disjoint from
the 94-class proof union and overlaps the 18-class counter-satisfiable union on
eight classes.  None of the three ATP runs is decisive on any of the 28 newly
classified identities.  These are bounded observations, not impossibility
results; the construction comparison concerns the 16 externally justified
infinite existences, rather than the aggregate 112 proof- and model-side
statuses.

\paragraph{The published ALPS baseline.}
Xie et al.'s experiment uses a different protocol: an eight-configuration
portfolio and a staged per-configuration sweep up to 600 seconds
~\cite{XieEtAl2026ALPS,ALPSData2026}.  At the 30-second rung it reports 87 new
implication proofs and four construction-side models; the later rungs add 23
implication proofs and no models.  Thus its largest-budget total is 110
proof-side results and four model-side results, leaving 4,027 of 4,141 rows
unresolved.  All four models come from \texttt{twee/complete}; the Vampire
saturation configurations' resolutions are implication proofs.  ALPS certifies
the four model-side cases through finite presentations and complete saturation,
but does not emit an explicit carrier operation as a Lean theorem.  We report
these published results separately and do not pool them with our same-resource
measurements.  They nevertheless reinforce the same qualitative point: the
scarce outcome is a constructed model, not a proof that an implication holds.

\subsection{Feasibility, environment, and artifact}

Table~\ref{tab:feasibility} reports the algorithmic interval from parsed input
through candidate compilation for accepted records, separate from Judge time.

\input{generated/feasibility_table.tex}

The accepted trace models are usually found early: the mean search time is
3.909 seconds, the median is 0.766 seconds, and the maximum is 118.562 seconds.
In contrast, the median
over all 4,187 Canonical-4187 trace runs is 120.406 seconds because an unsuccessful bounded
run normally consumes its allocation.  Peak RSS is sampled rather than read
atomically at the enforcing boundary; a few samples can therefore slightly
exceed 2,000 MiB before termination.  Judge time and queue or transport time
are recorded separately from search time.

The experiment used at most 120 reusable sandboxes and at most 48 concurrent
Judge calls.  The official solver image uses Python~3.11.  Receipts identify
Judge v3 revision \texttt{380d1b2c}, Lean~4.33.1, and Mathlib revision
\texttt{0df444a3}.  The artifact records the full combined service name,
the complete Mathlib commit, and the proof policy.

The release records both frozen source groups (4,271 rows), the 83 cross-group
matches defining Merged-4188, and the additional within-ALPS equivalence
defining Canonical-4187.  It also contains exact standalone programs for all
selected profiles, per-run results,
731 receipt occurrences, 731 individual certificate files, and a SHA-256
manifest.  Its verifier recomputes the fresh trace count, supplemental selected
coverage, exact intersections, certificate--receipt pairing, source hashes,
Judge revision, and privacy allowlist.  A separate same-resource ATP artifact
contains the completed run records, decisive proof streams, the
Canonical-4187 view, tool commands and versions, and its own manifest; its verifier recomputes
the benchmark-separated results, the canonical 94/18 unions, and the absence
of contradictory statuses.  The public
artifacts intentionally exclude controller secrets, sandbox identifiers,
private endpoints, and the full multi-method production solver.  The
repository is available at
\url{https://github.com/YanbiaoLab/trace-tree-magmas}.
Each experiment README gives a one-command offline verification path; no
network access or re-execution of the search is required for the evidence
audit.  Re-running the searches and replaying Lean proofs are documented as
separate reproducibility levels rather than being conflated with hash
verification.

\section{Related work}

Austin and Kisielewicz established and analyzed identities with no nontrivial
finite models~\cite{Austin1965,Austin1966,Kisielewicz1990,Kisielewicz1997}.
The present work addresses the constructive infinite half of that
classification problem.

\paragraph{Automatic finite and infinite model building.}
Finite model finders such as Mace4 enumerate bounded domains and are powerful
for ordinary equational countermodels, but a finite table cannot witness the
infinite side of an Austin law~\cite{McCune2003}.  Nitpick integrates a
relational model finder into Isabelle, supports rich higher-order syntax, and
obtains countermodels through finite scopes~\cite{BlanchetteNipkow2010}.
Recent Lean-oriented counterexample generation similarly produces checked
witnesses to false statements, but does not synthesize an infinite algebra
satisfying a universal magma identity~\cite{LiEtAl2026}.

Automated infinite-model construction itself has a substantial history.
Peltier represented infinite first-order models by tree automata and regular
tree grammars, later constructing non-ambiguous formulae for equational clause
sets and connecting finite and infinite model builders
~\cite{Peltier1997,Peltier2003,Peltier2008}.  Model Evolution lifts the DPLL
view to first-order logic and represents candidate interpretations by evolving
contexts~\cite{BaumgartnerTinelli2008}.  AGES synthesizes logical models for
order-sorted first-order theories from algebraic interpretations and has both
a formal account and a dedicated implementation
~\cite{GutierrezLucas2018,GutierrezLucas2019}.  More recently, symbolic model
construction has been developed for saturated constrained Horn clauses and
for first-order literal models~\cite{BrombergerEtAl2023,BrombergerEtAl2024},
while RInGen searches quantified invariants over algebraic data types using
grammar-based symbolic representations~\cite{KostyukovEtAl2021}.  These lines
make clear that neither automated infinite-model construction nor finite
descriptions of infinite interpretations are first introduced here.

Elad, Padon, and Shoham provide another genuinely infinite representation:
symbolic structures for satisfiable quantified verification conditions,
together with a model-finding procedure and a decidable fragment for which the
representation is complete~\cite{EladPadonShoham2024}.  Trace-tree magmas have
a different mathematical interface.  They search for one total binary
operation on a free constructor algebra.  A source-law trace proposes finitely
many Horn rules for exceptional outputs; descent on constructor size proves
that the resulting relation is functional; and exhaustive case analysis over
default pairs and exceptional outputs proves the universal identity.  Our
search space and proof requirements are therefore specialized to equational
algebra rather than being an alternative notation for those general
frameworks.

\paragraph{ATP saturation and extracted models.}
General first-order ATPs such as Vampire establish many equational implications
by refutation~\cite{KovacsVoronkov2013,Janota2025}.  A saturated unit-equational
clause set can sometimes be read as a convergent rewrite presentation of an
explicit, possibly infinite model.  Janota, Rawson, and Schulz systematize that
construction for Vampire and E and propose certified termination and
confluence checks~\cite{JanotaRawsonSchulz2026}.  Their reported study starts
from successful saturations in the order-at-most-four ETP corpus; it does not
cover the order-five identities classified here.  Our procedure instead
searches the constructor-tree operation directly and emits its source proof
and separating valuation in Lean.

Against this background, our claim is specific.  We introduce
\emph{rank-decreasing sparse trace-tree magmas} as a model language for magma
identities, an automated procedure that searches this language directly from a
source identity, and a compiler whose output is a self-contained Lean~4 proof
of functionality, universal source validity, carrier infinitude, and target
failure.

Deductive program synthesis also couples search and proof.  Hozzov\'a et al.
extend superposition with answer literals to synthesize recursion-free
programs from proofs of functional specifications~\cite{HozzovaEtAl2023}.
Our target is instead a semantic structure---a total operation on a provably
infinite carrier---whose universal law and countervaluation are Lean-certified;
our grammar and obligations therefore concern model construction rather than
program extraction.

Knuth--Bendix completion and term-rewriting theory motivate free constructors,
critical matches, and normal forms~\cite{KnuthBendix1970,BaaderNipkow1998}.
Sternagel and Winkler's verified certifier checks ordered-completion runs and
equational satisfiability certificates in Isabelle/HOL
~\cite{SternagelWinkler2019}.  That work validates completion-based reasoning;
our main procedure instead synthesizes an exceptional operation from a source
evaluation trace, while the smaller completion variant emits a direct Lean
model rather than an ordered-completion run.
The trace family is not merely an orientation of the source.  Its mutually
inductive admissible-step relation permits earlier exceptional evaluations,
and the generator must prove the exceptional-output relation functional.

ETP combines large-scale automated reasoning with formal verification
~\cite{BolanEtAl2025}, while ALPS explores construction-oriented mathematical
reasoning~\cite{XieEtAl2026ALPS}.  To the best of our literature and
public-artifact audit, no prior work defines or searches the trace-tree model
family above, or compiles discovered members of that family into self-contained
Lean~4 model certificates.  Our searches certify 32 order-five Austin
candidates in total.  For 28 of them---fourteen classes modulo opposite-magma
duality---they provide the first public infinite countermodels and hence 28 new
Austin classifications.  The remaining four identities, forming two duality
classes previously reported by ALPS, are independent reconstructions and are
explicitly excluded from that novelty count.

\section{Limitations and future work}

The trace grammar and its bounded enumeration are intentionally incomplete.  It is strongest for identities
orientable as a variable equal to a term and whose evaluation closes through a
small number of exceptional traces.  Failure to find a model proves nothing.
In particular, this work does not classify inputs on which a procedure finishes
without a certificate or is stopped by a time or memory limit.  A generated
candidate is likewise not evidence until Lean accepts its certificate.  Lean
accepts every certificate in the fresh trace run and in the guarded and CNF
regressions.  In the main completion regression, three generated certificates
receive no terminal response before the 60-minute client cutoff; accepted
indexed alternatives for the same three problems are supplied, but the original
three files are not counted as accepted.  A further generalized schedule covers
the main profile's one no-candidate case.  We make no correctness claim for an
unaccepted file, and the artifact preserves every outcome.

The one-clause result suggests a mathematical research program.  A syntactic
criterion for a constructor projection that recovers the output parameter,
together with an acyclic size ranking, could
replace mask enumeration by a theorem for a class of identities.  Multiple
compatible Horn clauses may cover terms with several essential branches.
Richer constructors or lexicographic ranks may expose models that the current
tree size cannot certify.  Finally, the infinite models for
\Law{17260}/\Law{28740} motivate a finite-model investigation to determine
whether this pair is Austin.

At the implementation level, a future verified compiler could reduce the
trusted engineering gap between the Python candidate and its Lean rendering.
The present design instead follows a proof-producing discipline: every output
is independently reconstructed and checked, so compiler bugs cause rejection
rather than unsound acceptance.

\section{Conclusion}

Trace-tree magmas provide a finite search language for genuinely infinite
algebraic countermodels.  A least relational semantics defines their meaning
independently of the implementation.  For each reported model, a generated
Lean proof uses constructor projection and size descent to establish
functionality of the exceptional-output relation.  Complete case analysis over
default pairs and exceptional outputs proves the source identity for every
valuation, and a concrete countervaluation refutes the target.  This theorem
chain turns heuristic search into proof-producing model synthesis without
claiming that the search is complete.

The method establishes the first public infinite countermodels in our audit
for 28 previously unclassified order-five identities, or fourteen classes
modulo duality.  Each new class admits a single dependency-respecting Horn
clause with at most two admissible-step premises, revealing a common
mathematical structure rather than 28 unrelated witnesses.  Evaluation keeps
Order5-130 and ALPS-4141 distinct: their 4,188-record deduplicated union has
4,187 classes after one further within-ALPS pair is identified under variable
renaming and equality-side exchange.  In a fresh
run on those classes, trace produces 636 certificates---36 on Order5-130 and
600 on the ALPS-only remainder---all accepted by the official Judge.  Of these,
624 canonical classes have ALPS representatives and correspond to 625 rows of
the original 4,141-row ALPS pool.  Targeted regressions separately confirm accepted outputs of the guarded,
completion, and CNF model languages without folding their union into the ATP
comparison.  In the same-resource experiment, every one of the 28 new cases
remains without a decisive ATP answer.  The three ATPs prove implications for
94 canonical classes; only complete Twee supplies trusted counter-satisfiable
results, on 18 classes, and independent finite-side facts force infinitude for
16 of them.  No ATP run produces an explicit Lean-certified infinite model.
These results
locate the contribution at the intersection of automated deduction, explicit
model synthesis, and proof production, and motivate a broader theory of which
magma laws admit finite trace-tree presentations.

\paragraph{Author contributions.}
Jiaming Zhao generalized the central concept, extracted the mathematical
structure, designed and implemented the algorithms and Lean certificate
generator, ran the large-scale experiments, discovered the new Austin-law
classifications, and wrote the manuscript.  Bing Wu developed and maintained
all experimental infrastructure, including the database,
\texttt{judge-v3-repl} distributed Lean verification service, and Alibaba
Cloud Sandbox setup, and contributed to offline exploration.  Xu Miao is the
corresponding author.

\appendix

\section{The 32 certified Austin identities}
\label{app:identities}

The operation is written as $\ast$.  The first 28 rows are the new
classifications.  The final two dual pairs,
\Law{19966}/\Law{26105} and \Law{22619}/\Law{22634}, are independent formal
reconstructions of earlier ALPS classifications.

\input{generated/equations_table.tex}

\bibliographystyle{splncs04}
\bibliography{references}

\end{document}

%% file: generated/trace_schemas_table.tex
\begin{table}[!t]
\caption{The fourteen new duality classes as one-clause trace-tree models.
Here $\langle a,b\rangle$ is the free pair constructor, $\mathsf S$ is the
admissible-step relation, and $\mathsf C$ is the exceptional-output relation.  The
listed representative uses $\star=\mathsf{eval}$; the other law in the first
column is modeled by the opposite magma.}
\label{tab:trace-schemas}
\centering
\scriptsize
\setlength{\tabcolsep}{3pt}
\begin{tabular}{@{}>{\raggedright\arraybackslash}p{3.0cm}c>{\raggedright\arraybackslash}p{7.25cm}@{}}
\toprule
Duality class & Premises & Law-specific Horn clause \\
\midrule
\Law{4952} / \Law{41252}
  & 1 & $\mathsf S(z,y;h)\Rightarrow
          \mathsf C(y,\langle x,\langle y,\langle y,h\rangle\rangle\rangle;x)$ \\
\Law{40914} / \Law{4957}
  & 1 & $\mathsf S(y,x;h)\Rightarrow
          \mathsf C(\langle\langle\langle h,y\rangle,z\rangle,x\rangle,z;x)$ \\
\Law{5012} / \Law{41253}
  & 1 & $\mathsf S(y,z;h)\Rightarrow
          \mathsf C(y,\langle x,\langle z,\langle z,h\rangle\rangle\rangle;x)$ \\
\Law{5066} / \Law{41239}
  & 1 & $\mathsf S(z,y;h)\Rightarrow
          \mathsf C(y,\langle y,\langle x,\langle y,h\rangle\rangle\rangle;x)$ \\
\Law{5141} / \Law{40917}
  & 1 & $\mathsf S(x,y;h)\Rightarrow
          \mathsf C(y,\langle y,\langle z,\langle y,h\rangle\rangle\rangle;x)$ \\
\Law{5295} / \Law{40909}
  & 1 & $\mathsf S(x,y;h)\Rightarrow
          \mathsf C(y,\langle z,\langle y,\langle y,h\rangle\rangle\rangle;x)$ \\
\Law{38303} / \Law{7701}
  & 2 & $\mathsf S(x,z;h_0)\wedge
          \mathsf S(y,\langle h_0,x\rangle;h_1)\Rightarrow
          \mathsf C(\langle h_1,y\rangle,y;x)$ \\
\Law{7755} / \Law{38249}
  & 1 & $\mathsf S(\langle z,\langle x,x\rangle\rangle,y;h)\Rightarrow
          \mathsf C(y,\langle y,h\rangle;x)$ \\
\Law{9345} / \Law{36713}
  & 1 & $\mathsf S(x,y;h)\Rightarrow
          \mathsf C(y,\langle h,\langle z,\langle y,y\rangle\rangle\rangle;x)$ \\
\Law{9384} / \Law{36714}
  & 1 & $\mathsf S(x,z;h)\Rightarrow
          \mathsf C(y,\langle h,\langle y,\langle z,z\rangle\rangle\rangle;x)$ \\
\Law{9680} / \Law{36524}
  & 2 & $\mathsf S(z,y;h_0)\wedge\mathsf S(x,y;h_1)\Rightarrow
          \mathsf C(y,\langle h_0,\langle y,h_1\rangle\rangle;x)$ \\
\Law{9667} / \Law{36638}
  & 1 & $\mathsf S(z,y;h)\Rightarrow
          \mathsf C(y,\langle h,\langle x,\langle y,y\rangle\rangle\rangle;x)$ \\
\Law{35036} / \Law{11081}
  & 2 & $\mathsf S(y,z;h_0)\wedge\mathsf S(x,y;h_1)\Rightarrow
          \mathsf C(\langle h_0,\langle h_1,x\rangle\rangle,y;x)$ \\
\Law{11116} / \Law{34888}
  & 1 & $\mathsf S(z,x;h)\Rightarrow
          \mathsf C(y,\langle\langle x,h\rangle,\langle y,y\rangle\rangle;x)$ \\
\bottomrule
\end{tabular}
\end{table}

%% file: generated/trace_full_results_table.tex
\begin{table}[t]
\caption{Fresh trace result on Canonical-4187.  The first three rows are the
disjoint constituents of Order5-130; ``ALPS-only'' is the residual canonical
contribution of ALPS-4141 after cross-group deduplication.  The three terminal
generation states are exhaustive.  ``No candidate'' and ``memory limit'' make
no mathematical claim about the input.}
\label{tab:trace-full-results}
\centering
\scriptsize
\setlength{\tabcolsep}{3pt}
\begin{tabular}{@{}lrrrr@{}}
\toprule
Canonical partition & Inputs & Accepted & No candidate & Memory limit \\
\midrule
Candidate-96 & 96 & 28 & 65 & 3 \\
KnownAustin-10 & 10 & 6 & 4 & 0 \\
Open-24 & 24 & 2 & 22 & 0 \\
\cmidrule(lr){1-5}
Order5-130 subtotal & 130 & 36 & 91 & 3 \\
ALPS-only canonical & 4,057 & 600 & 3,204 & 253 \\
\midrule
Canonical-4187 & 4,187 & 636 & 3,295 & 256 \\
\bottomrule
\end{tabular}
\end{table}

%% file: generated/isolated_results_table.tex
\begin{table}[t]
\caption{Scope-correct accounting for the four tuned searches.  Trace was run
on Canonical-4187, the canonical union of Order5-130 and ALPS-4141.  The other rows are regressions on each
method's historical hit set and therefore do not measure accuracy over all of
Canonical-4187.  ``Selected coverage'' uses accepted alternate completion
certificates for the four cases not accepted by the main profile.}
\label{tab:isolated-results}
\centering
\footnotesize
\setlength{\tabcolsep}{4pt}
\renewcommand{\arraystretch}{1.08}
\begin{tabular}{@{}lrrrr@{}}
\toprule
Search and evaluation scope & Tested & \shortstack{Main\\accepted} &
\shortstack{Selected\\accepted} & \shortstack{Trace\\overlap} \\
\midrule
Trace depth sweep (Canonical-4187) & 4,187 & 636 & 636 & -- \\
Guarded decoder (historical hits) & 2 & 2 & 2 & 0 \\
Strict completion (historical hits) & 61 & 57 & 61 & 38 \\
CNF NF16 (sound historical hits) & 17 & 17 & 17 & 8 \\
\bottomrule
\end{tabular}
\end{table}

%% file: generated/additional_models_table.tex
\begin{table}
\caption{Supplementary Order5-130 results on six further duality classes.
The first five rows lie in KnownAustin-10; the last lies in Open-24.  Trace
accepts three known Austin classes and the class whose two statuses were open;
completion covers five classes and CNF covers all six in their targeted
regressions.  The last row now has a kernel-checked infinite model, but its
finite-model status remains open.}
\label{tab:additional-models}
\centering
\scriptsize
\setlength{\tabcolsep}{3pt}
\begin{tabular}{@{}llccc@{}}
\toprule
Duality class & Prior status & Trace & Completion & CNF rules \\
\midrule
\Law{4916} / \Law{41082} & known Austin & -- & accepted & 5 \\
\Law{15535} / \Law{30591} & known Austin & -- & -- & 13 \\
\Law{17522} / \Law{28770} & known Austin & accepted & accepted & 2 \\
\Law{20034} / \Law{25964} & known Austin & accepted & accepted & 1 \\
\Law{22455} / \Law{22818} & known Austin & accepted & accepted & 2 \\
\Law{17260} / \Law{28740} & both statuses open & accepted & accepted & 2 \\
\bottomrule
\end{tabular}
\end{table}

%% file: generated/same_resource_atp_table.tex
\begin{table}[t]
\caption{Same-resource ATP outcomes, with the overlapping source groups kept
separate.  ``Proofs'' are implications to \Law{2}, not models.  A trusted
counter-satisfiable result is available only from complete Twee; ``infinite''
is then inferred from an independent finite-side certificate, not emitted by
the ATP.  The four ALPS counter rows duplicate Order5-130 rows.  No ATP emitted
an explicit model or a Lean-checkable certificate.}
\label{tab:same-resource-atp}
\centering
\scriptsize
\setlength{\tabcolsep}{2.5pt}
\begin{tabular}{@{}llrrrr@{}}
\toprule
Benchmark & Tool & \shortstack{Implication\\proofs} &
\shortstack{Trusted\\counter-sat.} &
\shortstack{Infinite by\\admissibility} &
\shortstack{No\\decisive result} \\
\midrule
Order5-130 & Vampire 5.0.1 & 0 & 0 & 0 & 130 \\
Order5-130 & E 3.5.1 & 0 & 0 & 0 & 130 \\
Order5-130 & Twee complete & 2 & 18 & 16 & 110 \\
\midrule
ALPS-4141 & Vampire 5.0.1 & 28 & 0 & 0 & 4,113 \\
ALPS-4141 & E 3.5.1 & 54 & 0 & 0 & 4,087 \\
ALPS-4141 & Twee complete & 79 & 4 & 4 & 4,058 \\
\bottomrule
\end{tabular}
\end{table}

%% file: generated/feasibility_table.tex
\begin{table}[t]
\caption{Search time and sampled peak process-tree RSS.  Entries are
mean / median / maximum; Judge time is excluded.  Trace rows come from the
fresh full run.  Auxiliary rows are measured only on the targeted regressions
of Table~\ref{tab:isolated-results}.}
\label{tab:feasibility}
\centering
\scriptsize
\setlength{\tabcolsep}{3pt}
\renewcommand{\arraystretch}{1.08}
\begin{tabular}{@{}lrrrrrrr@{}}
\toprule
Profile / scope & Count & \multicolumn{3}{c}{Search time (s)} &
\multicolumn{3}{c}{Peak RSS (MiB)} \\
\cmidrule(lr){3-5}\cmidrule(l){6-8}
& & Mean & Median & Max. & Mean & Median & Max. \\
\midrule
Trace / Candidate-96 accepted & 28 & 1.29 & 0.58 & 10.77 & 112.8 & 66.1 & 823.3 \\
Trace / KnownAustin-10 accepted & 6 & 0.54 & 0.57 & 0.64 & 64.9 & 66.0 & 68.9 \\
Trace / Open-24 accepted & 2 & 0.55 & 0.55 & 0.56 & 64.2 & 64.2 & 66.2 \\
Trace / ALPS-only accepted & 600 & 4.08 & 0.79 & 118.56 & 163.6 & 66.1 & 1908.3 \\
Trace / Canonical-4187 accepted & 636 & 3.91 & 0.77 & 118.56 & 160.2 & 66.1 & 1908.3 \\
Guarded / targeted accepted & 2 & 0.58 & 0.58 & 0.59 & 41.3 & 41.3 & 41.3 \\
Completion (main) / targeted accepted & 57 & 13.95 & 13.64 & 43.00 & 122.9 & 75.3 & 1934.8 \\
CNF / targeted accepted & 17 & 1.15 & 1.22 & 1.73 & 22.7 & 23.9 & 24.5 \\
\bottomrule
\end{tabular}
\end{table}

%% file: generated/equations_table.tex
\begin{center}
\refstepcounter{table}\label{tab:equations}
\parbox{\textwidth}{\small\textbf{Table \thetable.} The 32 certified Austin
identities, grouped by duality.  The binary operation is written as $\ast$.
The first 28 identities are new classifications established by trace search.
The final two pairs were previously classified by ALPS and are independently
reconstructed here by the supplemental searches.}
\medskip
\scriptsize
\setlength{\tabcolsep}{4pt}
\begin{tabular}{@{}>{\raggedright\arraybackslash}p{2.3cm}>{\raggedright\arraybackslash}p{8.3cm}@{}}
\toprule
Law & Identity \\
\midrule
\Law{4952} & $x \simeq y \ast (x \ast (y \ast (y \ast (z \ast y))))$ \\
\Law{41252} & $x \simeq ((((y \ast z) \ast y) \ast y) \ast x) \ast y$ \\
\addlinespace
\Law{40914} & $x \simeq ((((y \ast x) \ast y) \ast z) \ast x) \ast z$ \\
\Law{4957} & $x \simeq y \ast (x \ast (y \ast (z \ast (x \ast z))))$ \\
\addlinespace
\Law{5012} & $x \simeq y \ast (x \ast (z \ast (z \ast (y \ast z))))$ \\
\Law{41253} & $x \simeq ((((y \ast z) \ast y) \ast y) \ast x) \ast z$ \\
\addlinespace
\Law{5066} & $x \simeq y \ast (y \ast (x \ast (y \ast (z \ast y))))$ \\
\Law{41239} & $x \simeq ((((y \ast z) \ast y) \ast x) \ast y) \ast y$ \\
\addlinespace
\Law{5141} & $x \simeq y \ast (y \ast (z \ast (y \ast (x \ast y))))$ \\
\Law{40917} & $x \simeq ((((y \ast x) \ast y) \ast z) \ast y) \ast y$ \\
\addlinespace
\Law{5295} & $x \simeq y \ast (z \ast (y \ast (y \ast (x \ast y))))$ \\
\Law{40909} & $x \simeq ((((y \ast x) \ast y) \ast y) \ast z) \ast y$ \\
\addlinespace
\Law{38303} & $x \simeq ((y \ast ((x \ast z) \ast x)) \ast y) \ast y$ \\
\Law{7701} & $x \simeq y \ast (y \ast ((x \ast (z \ast x)) \ast y))$ \\
\bottomrule
\end{tabular}
\medskip

\begin{tabular}{@{}>{\raggedright\arraybackslash}p{2.3cm}>{\raggedright\arraybackslash}p{8.3cm}@{}}
\toprule
Law & Identity \\
\midrule
\Law{7755} & $x \simeq y \ast (y \ast ((z \ast (x \ast x)) \ast y))$ \\
\Law{38249} & $x \simeq ((y \ast ((x \ast x) \ast z)) \ast y) \ast y$ \\
\addlinespace
\Law{9345} & $x \simeq y \ast ((x \ast y) \ast (z \ast (y \ast y)))$ \\
\Law{36713} & $x \simeq (((y \ast y) \ast z) \ast (y \ast x)) \ast y$ \\
\addlinespace
\Law{9384} & $x \simeq y \ast ((x \ast z) \ast (y \ast (z \ast z)))$ \\
\Law{36714} & $x \simeq (((y \ast y) \ast z) \ast (y \ast x)) \ast z$ \\
\addlinespace
\Law{9680} & $x \simeq y \ast ((z \ast y) \ast (y \ast (x \ast y)))$ \\
\Law{36524} & $x \simeq (((y \ast x) \ast y) \ast (y \ast z)) \ast y$ \\
\addlinespace
\Law{9667} & $x \simeq y \ast ((z \ast y) \ast (x \ast (y \ast y)))$ \\
\Law{36638} & $x \simeq (((y \ast y) \ast x) \ast (y \ast z)) \ast y$ \\
\addlinespace
\Law{35036} & $x \simeq ((y \ast z) \ast ((x \ast y) \ast x)) \ast y$ \\
\Law{11081} & $x \simeq y \ast ((x \ast (y \ast x)) \ast (z \ast y))$ \\
\addlinespace
\Law{11116} & $x \simeq y \ast ((x \ast (z \ast x)) \ast (y \ast y))$ \\
\Law{34888} & $x \simeq ((y \ast y) \ast ((x \ast z) \ast x)) \ast y$ \\
\addlinespace
\Law{19966} & $x \simeq (y \ast y) \ast ((x \ast (x \ast z)) \ast z)$ \\
\Law{26105} & $x \simeq (y \ast ((y \ast x) \ast x)) \ast (z \ast z)$ \\
\addlinespace
\Law{22619} & $x \simeq (y \ast (y \ast x)) \ast ((z \ast z) \ast z)$ \\
\Law{22634} & $x \simeq (y \ast (y \ast y)) \ast ((x \ast z) \ast z)$ \\
\bottomrule
\end{tabular}
\end{center}